\documentclass[12pt]{article}

\usepackage{amsmath, amssymb, amsfonts, amsbsy, amsthm, latexsym, color}
\usepackage{graphicx}
\usepackage{enumerate}
\everymath{\displaystyle}
\newtheorem{theorem}{Theorem}[section]
\newtheorem{proposition}[theorem]{Proposition}
\newtheorem{lemma}[theorem]{Lemma}
\newtheorem{corollary}[theorem]{Corollary}

\newtheorem{definition}[theorem]{Definition}

\newtheorem{remark}[theorem]{Remark}

\numberwithin{equation}{section}

\newcommand{\C}{{\mathbb{C}}}

\newcommand{\Nc}{\mathcal{N}}
\newcommand{\R}{{\mathbb{R}}}
\newcommand{\Sb}{\mathbb{S}}

\newcommand{\Z}{{\mathbb{Z}}}

\makeatletter
\newcommand{\subjclass}[2][]{%
  \let\@oldtitle\@title%
  \gdef\@title{\@oldtitle\footnotetext{#1 \emph{MSC.} #2}}%
}
\newcommand{\keywords}[1]{%
  \let\@@oldtitle\@title%
  \gdef\@title{\@@oldtitle\footnotetext{\emph{Keywords:} #1.}}%
}
\makeatother

\begin{document}

\title{The gap between the null space property and the restricted isometry property}

\author{Jameson Cahill\footnote{Department of Mathematics,
New Mexico State University, Las Cruces, NM, USA. } \qquad Xuemei Chen\footnote{Department of Mathematics,  New Mexico State University, Las Cruces, NM, USA. X. Chen was supported by NSF ATD 1321779.}\qquad Rongrong Wang\footnote{Department of Mathematics, University of British Columbia, Vancouver, BC, Canada. 
R. Wang is funded in part by an NSERC
Collaborative Research and Development Grant DNOISE II (22R07504).}}

%
%
%


\date{\today}

\maketitle
\begin{abstract}
The null space property (NSP) and the restricted isometry property (RIP) are two properties which have received considerable attention in the compressed sensing literature.  As the name suggests, NSP is a property that depends solely on the null space of the measurement procedure and as such, any two matrices which have the same null space will have NSP if either one of them does.  On the other hand, RIP is a property of the measurement procedure itself, and given an RIP matrix it is straightforward to construct another matrix with the same null space that is not RIP.  
We say a matrix is RIP-NSP if it has the same null space as an RIP matrix. We show that such matrices can provide robust recovery of compressible signals under Basis pursuit. More importantly, we constructively show that the RIP-NSP is stronger than NSP with the aid of this robust recovery result, which shows that RIP is fundamentally stronger than NSP.
\end{abstract}

\vspace{0.35in}
\section{Introduction}
Let $x\in\R^N$ be an unknown signal that belongs to or lives close to the set $\Sigma_s:=\{w\in\R^N:\|w\|_0:=\#(\text{supp(w))}\leq s\}$ in the sense that $\sigma_s(x):=\min_{y\in\Sigma_s}\|x-y\|_1$ is small. The conventional \emph{compressed sensing} problem is concerned with the task of estimating $x$ provided that we know the sensing procedure $\Phi\in\R^{M\times N}$, and we are given the (possibly perturbed) measurement data $y=\Phi x+e$ with $\|e\|_2\leq\epsilon$. The challenging part is that the number of measurements $M$ is far less than the ambient dimension $N$. In this paper, we focus on the $\ell_1$ minimization method (Basis Pursuit): 
\begin{equation}\label{equ:le}
\hat x=\arg\min\|\bar x\|_1\quad\quad \text{subject to }\quad\quad \|\Phi\bar x -y\|_2\leq\epsilon,
\end{equation}
which is a well established reconstruction scheme \cite{CT05}.

A significant portion of the compressed sensing literature is devoted to the design of the sensing matrix $\Phi$ to guarantee the success of \eqref{equ:le}. More specifically, we want the recovery to be stable and robust in the following way:
\begin{equation}\label{equ:stable}
\|\hat x-x\|_2\leq C_1\epsilon+C_2 \sigma_s(x),
\end{equation}
where $C_1, C_2$ are constants that do not grow as dimension grows.

 The restricted isometry property (RIP) and the null space property (NSP) (see definitions \ref{def:nsp} and \ref{def:RIP} respectively) are two common conditions that one imposes on $\Phi$ in order to get a stable recovery via \eqref{equ:le}. The RIP condition with the RIP constant $\delta_{2s}<1/\sqrt{2}$ is sufficient to recover $x$ (with stability and robustness as in \eqref{equ:stable}, see Theorem \ref{thm:cai}), so in what follows, when we say RIP, we really mean RIP with $\delta_{2s}\leq 1/\sqrt{2}$. 
 
The two conditions have their pros and cons:

\begin{enumerate}
\item[(1)] NSP is a property that only depends on the kernel of the sensing matrix $\Phi$, and RIP is not.
\item[(2)] $s$-NSP is a sufficient and necessary condition for any $x\in\Sigma_s$ to be exactly recovered by \eqref{equ:le} in the noise free case (when $\epsilon=0$).
\item[(3)] RIP provides a very stable and robust error bound as we will see in Theorem \ref{thm:cai}. It is not known that NSP guarantees such recovery.
\item[(4)] Numerically verifying  RIP and NSP  are both NP-hard, but RIP is easier for theoretical justification in most cases. RIP (or RIP like conditions)  also has a broader application to various sparse recovery algorithms.

\end{enumerate}

Note that the solution set of $\Phi x=y$ only depends on the null space of $\Phi$, and any row action does not alter this set. Practically, the row operation corresponds to a mixing or transformation of the current data vector. If this transformation  is invertible and the measurement is noiseless, then no new data is introduced or no old data is lost, and one should expect the same recovery. As such, it is natural to expect that a good reconstruction criteria should be able to reflect this invariant property. Moreover, it is also desirable to have a property that guarantees robust recovery as in \eqref{equ:stable}, so that the Basis Pursuit is not affected significantly by noisy measurements, and is able to recover compressible  signals (signals that have small $\sigma_s(x)$) approximately. In other words, we wish to have a property that combines the advantages of RIP and NSP so that:
\begin{enumerate}
\item[(A)] This property only depends on the null space of the sensing matrix.
\item[(B)] This property guarantees stable and robust recovery similar to \eqref{equ:stable}.
\end{enumerate}

There have been several attempts in the literature: the width property (WP) \cite{KT07}, the robust null space property (RNSP)~\cite{F14}, and the robust width property (RWP)~\cite{CM14} are all properties on  sensing matrices which guarantee the success of Basis Pursuit. However, none of these properties simultaneously achieve both (A) and (B). A more elaborated discussion about these properties can be found in Section \ref{sec:ripnsp}. 

The first contribution of this paper is to provide a condition that only depends on the null space (thus achieves (A)), and also achieves (B) under certain settings.  The condition we propose is a straightforward hybrid of RIP and NSP:
\begin{definition}[$k$-RIP-NSP]
A matrix $\Phi$ has $k$-RIP-NSP with $\delta_k$ if $\Phi$ has the same null space of a  matrix that has RIP with $\delta_k$.
\end{definition}

\noindent Which provides the following guarantee for Basis Pursuit:
\begin{theorem}\label{thm:sta} Suppose $\Phi$ is $M\times N$ and has RIP-NSP with $\delta_{2s}<\frac{1}{\sqrt 2}$. Given $x\in\R^N$, and $\|y-\Phi x\|_2\leq\epsilon$, then the minimizer $\hat x$ of \eqref{equ:le} satisfies
\begin{equation}\label{equ:ripnsp}
\|\hat x-x\|_2\leq D_1\frac{\sqrt{1+\delta_{2s}}}{\lambda(\Phi)}\sqrt{\frac{N}{s}}\epsilon+D_2\frac{\sigma_s(x)}{\sqrt{s}},
\end{equation}
where $\lambda(\Phi)$ is the smallest positive singular value of $\Phi$, and $D_1, D_2$ only depend on $\delta_{2s}$ and can be found in Theorem \ref{thm:cai}.
\end{theorem}

Admittedly, the stability result is not as strong as \eqref{equ:stable}. Since the RIP-NSP condition has no control of the scaling nor the condition number of $\Phi$, we do not expect a substantial improvement over \eqref{equ:ripnsp}.  However, we would like to argue in certain settings this is a reasonable constant.
In the RIP regime, say if $\Phi$ is a  Gaussian matrix normalized by $\sqrt M$, the term containing $D_1$ in \eqref{equ:ripnsp} becomes of order $\sqrt{M/s}$ with high probability  which is only worse than \eqref{equ:stable} by a $\log N$ factor for the optimal dimension $M=O(s\log N)$. Moreover, in some settings $N/s$ is a fixed constant, 
see for example \cite{BCTT11}.


 Studying sensing matrices with the RIP-NSP property have potential applications including topics related to preconditioning of sensing matrices. For example,  as the authors of \cite{KW13} point out to us, in compressive imaging (wavelet sparse images,  Fourier measurements), the variable-density sampling matrix has the null space of an RIP matrix but does not have RIP. 


People have always wondered whether the RIP condition is too strong for compressed sensing. In fact, one motivation for proposing  RNSP and RWP is to have a weaker-than-RIP condition, and yet provides nice theoretical performance. Here we want to raise an even more basic question: Is RIP strictly stronger than NSP? With a little argument, we can see this question is already solved with a positive answer. The set of matrices $\{U\Phi: U \text{ is invertible}\}$ share the same null space, but can have dramatically different RIP constants. For example, suppose $\Phi$ has a small $\delta_2$, which means every two columns of $\Phi$ are nearly orthogonal to each other. One can easily find an invertible matrix $U$ that sends the first two columns of $\Phi$ to be nearly parallel, which will result in a big constant $\delta_2$ for $U\Phi$. 

In order to do a fair comparison between RIP and NSP, we need to mod out this left multiplication, so the right question to ask is: Given an NSP matrix $\Phi$, is there an RIP matrix that belongs to the set $\{U\Phi: U \text{ is invertible}\}$? This is exactly asking whether RIP-NSP is strictly stronger than NSP, which gives a second purpose to this new condition.

The second contribution of this paper is to show that RIP-NSP is strictly stronger than NSP.

\begin{theorem}\label{thm:no}
Given arbitrary $\gamma<1$ and $\delta_{2s}<\frac{1}{\sqrt{2}}$,
if $M,N,s\in\Z$ satisfy that  $N/2\geq M\geq 2C(\gamma/3)s\log^4N, N/s\geq11$, and $s\geq \frac{640(1+\delta_{2s})^2}{(1-\sqrt{2}\delta_{2s})^2\gamma^2}$, then there exists a sensing matrix $\Phi\in\R^{M\times N}$ that satisfies $(s,\gamma)$-NSP, but does not satisfy RIP-NSP with $\delta_{2s}$.
\end{theorem}

This means that the null space property is fundamentally weaker than the restricted isometry property as it is possible to find a subspace $V$ that satisfies the NSP, but no matrix with kernel $V$ has RIP. Moreover, this theorem is surprisingly strong in the sense that no matter how small $\gamma$ is, NSP cannot be as good as RIP-NSP.

Although RIP-NSP achieves (A)-(B) (when $N/s$ is fixed), the RIP-NSP property has the drawback that it is not in the form of a restriction directly on the vectors in the null space.  
For future work, we wish to find an equivalent condition of RIP-NSP that is in a more direct form. This is also a question of whether one can construct a sensing matrix with restricted isometry property given that its null space is fixed and nicely structured. Theorem \ref{thm:no} suggests that we need a property stronger than NSP. 
For the rest of the paper, Section \ref{sec:ripnsp} provides an overview of RIP, NSP and other related properties, and Section \ref{sec:proof} provides proofs of Theorem \ref{thm:sta} and Theorem \ref{thm:no}.

\section{RIP, NSP, and more}\label{sec:ripnsp}

We will first review and define some notations for ease of the reading.
For a vector $v\in\R^N$ and $p>0$, let $\|v\|_p=\left(\sum_{i=1}^N|v_i|^p\right)^{1/p}$. If $I$ is an index set, $I^c$ is the complement set of $I$, and $v_I\in\R^N$ is the vector that has the same entry as $v$ on $I$, but 0 everywhere else. However, for convenience, we also use $v_I$ to indicate the truncation of $v$ on $I$, in which case $v_I\in \R^{|I|}$. It will be clear which interpretation is used from context.

For a matrix $M$, $\lambda(M)$ is the smallest positive singular value of $M$, and $\|M\|_{op}$ is the operator norm of $M$. 

\begin{definition}[$s$-NSP]\label{def:nsp}
A matrix $A$ has the null space property of order $s$ if 
\begin{equation}\label{equ:snsp}
\forall\ v\in\ker \Phi\backslash\{0\}, \forall\ |T|\leq s, \|v_T\|_1< \|v_{T^c}\|_1.
\end{equation}
\end{definition}

As summarized above, the null space property is known to be the necessary and sufficient condition of successful reconstruction of any $s$ sparse signals via the Basis pursuit when no noise is present~\cite{GN03}.

With an argument using compactness of $\ker(\Phi)\cap \Sb^{d-1}$, where $\Sb^{d-1}$ is the unit sphere in $\R^d$, one can prove that \eqref{equ:snsp} is equivalent to the existence of $0<\gamma<1$ such that
\begin{equation}\label{equ:gamma}
\forall\ v\in\ker \Phi, \forall\ |T|\leq s, \|v_T\|_1\leq\gamma \|v_{T^c}\|_1.
\end{equation}
In fact, the null space property was also proposed in this format~\cite{CDD09}.
The smallest $\gamma$ that makes \eqref{equ:gamma} holds is called the {\em null space constant}. We use $(s,\gamma)$-NSP to denote the null space property of order $s$ with null space constant $\gamma$. We will also abuse the definition by saying a subspace $V$ has $(s,\gamma)$-NSP if $\|v_T\|_1\leq\gamma \|v_{T^c}\|_1$ for every $v\in V$ and every $|T|\leq s$.

The null space property proposed has a recovery guarantee with noisy measurements, but it does not have a bound like in \eqref{equ:stable}.

\begin{theorem}\cite[Theorem III.4.1]{chen}
If $\Phi$ has $(s,\gamma)$-NSP, then any minimizer $\hat x$ of \eqref{equ:le} satisfies
\begin{equation}\label{equ:nsp}
\|\hat x-x\|_1\leq \frac{4\sqrt{2}\sqrt N }{(1-\gamma)\lambda(\Phi)}\epsilon+\frac{4(1+\gamma)}{\sqrt{2}(1-\gamma)}\sigma_s(x).
\end{equation}
\end{theorem}

More results related to the stability of $\ell_q$ version of null space property can be found in \cite{S12}. The recent work by Chen et al~\cite{CWW14} generalizes the null space property to the dictionary case where a similar stability bound is also proved.

The fundamental work on compressed sensing by Candes \cite{RIP08} admits a \eqref{equ:stable} type stability result if $\Phi$ satisfies the restricted isometry property. But we will list an optimal and more updated result by Cai and Zhang \cite{CZ14} here.

\begin{definition}[RIP]\label{def:RIP}
A matrix $\Phi$ has the restricted isometry property of sparsity $k$ if there exists $0<\delta<1$ such that
\begin{equation*}
(1-\delta)\|v\|_2^2\leq\|\Phi v\|_2^2\leq (1+\delta)\|v\|_2^2,
\end{equation*}
for all $v\in\Sigma_k$.

Moreover, the smallest $\delta$ that satisfies the above inequality is called the restricted isometry constant, and denoted as $\delta_k(\Phi)$ or $\delta_k$ when not ambiguous.
\end{definition}
\begin{theorem}[Cai and Zhang~\cite{CZ14}]\label{thm:cai}
If $\delta_{2s}(\Phi)<\frac{1}{\sqrt{2}}$, then the minimizer $\hat x$ of \eqref{equ:le} satisfies
\begin{equation}\label{equ:ripbound}
\|\hat x-x\|_2\leq D_1\epsilon+D_2\frac{\sigma_s(x)}{\sqrt{s}},
\end{equation}
where $D_1=\frac{2\sqrt{2(1+\delta_{2s})}}{1-\sqrt{2}\delta_{2s}}$, and $D_2=2\left(\frac{\delta_{2s}+\sqrt{\delta_{2s}(1/\sqrt{2}-\delta_{2s})}}{\sqrt{2}(1/\sqrt{2}-\delta_{2s})}+1\right)$, see \cite[Theorem 2.1]{CZ14}.
\end{theorem}

Moreover, this RIP bound is optimal, i.e., for every $\epsilon>0$, there exists $\Phi$ with $\delta_{2s}(\Phi)<\frac{1}{\sqrt{2}}+\epsilon$, and a signal $x$ such that \eqref{equ:ripbound} does not hold \cite{DG09}.

The robust null space property is a stronger version of the null space property, but it is not a property that only depends on the kernel. A matrix $\Phi$ satisfies the RNSP of order $s$ with constant $\rho, \tau>0$ if
$$\|v_T\|_2\leq\frac{\rho}{\sqrt{s}}\|v_{T^c}\|_1+\tau \|\Phi x\|_2, \text{ for all $v$ and all } |T|\leq s.$$ 
It is shown in \cite{F14} that RNSP is an equivalent condition to have a stability result like \eqref{equ:ripbound}, with the two constants only depend on $\rho, \tau$. Moreover, the author shows that RNSP is strictly weaker than RIP using the Weibull matrices.  The work \cite{DLR15,LM14} also describe a large class of random matrices which satisfy an RNSP but not an RIP condition. The sparse approximation property in \cite{S11} is similar to RNSP. 

In a recent paper by Cahill and Mixon~\cite{CM14}, the authors propose a robust width property, which requires 
$$\|v\|_2\leq\frac{c_0}{\sqrt{s}}\|v\|_1, \quad\forall x \text{ such that } \|\Phi x\|_2\leq c_1\|x\|_2,$$ for some $c_0, c_1>0$. The novelty of RWP is that it is a necessary and sufficient condition not only for the stability of Basis pursuit, but also for nuclear norm minimization and possibly other optimization problems.


\section{Proofs}\label{sec:proof}


\subsection{Stability result with RIP-NSP}
To prove Theorem \ref{thm:sta}, we show something more general, which requires a Lemma.

\begin{lemma}\label{lem:op}
If a matrix $\Phi$ has RIP with $\delta_k$, then
\begin{equation}\label{equ:boundphi}
\|\Phi\|_{op}\leq \sqrt{N/k+1}\sqrt{1+\delta_k}
\end{equation}
\end{lemma}
\begin{proof}
Partition the index set $\{1,...,N\}$ into $\{1,...,N\}=\cup_{i=1}^m T_i$, where each $T_i$ is of size $k$ except for the last one $T_m$, which has cardinality no bigger than $k$. $T_i$'s are nonoverlaping, so $m=\lceil N/k\rceil$, which is the smallest integer greater than or equal to $N/k$.
\begin{align*}
\|\Phi\|_{op}&=\max_{\|x\|_2=1}\|\Phi x\|_2=\max_{\|x\|_2=1}\left\|\sum_{i=1}^m \Phi_{T_i}x_{T_i}\right\|_2\leq\max_{\|x\|_2=1}\sum_{i=1}^m\sqrt{1+\delta_k}\|x_{T_i}\|_2\\
&\leq \sqrt{1+\delta_k}\max_{\|x\|_2=1}\sqrt{m\sum_{i=1}^m \|x_{T_i}\|_2^2}=\sqrt{m}\sqrt{1+\delta_k}\leq \sqrt{N/k+1}\sqrt{1+\delta_k}.
\end{align*}
\end{proof}
\begin{proposition}\label{pro:main} Given $x,\hat x\in\R^N$ and a list of statements

(a) The sensing matrix $A$ satisfies RIP with $\delta_{2s}$;

(b) $\|\hat x\|_1\leq\|x\|_1$;

(c) $\|A x-A\hat x\|_2\leq 2\rho$;

(d) $\|\hat x-x\|_*\leq C_1\rho+C_2\sigma_s(x)$, where $\|\cdot\|_*$ is some norm (or quasinorm);

(e) The sensing matrix $\Phi$ has RIP-NSP with $\delta_{2s}$;

(f) $\|\Phi x-\Phi\hat x\|_2\leq 2\epsilon$;

(g) $\|\hat x-x\|_*\leq C_1\frac{\sqrt{1+\delta_{2s}}}{\lambda(\Phi)}\sqrt{\frac{N}{s}}\epsilon+C_2\sigma_s(x)$;

If (a)+(b)+(c) $\Rightarrow$ (d), then (e)+(b)+(f) $\Rightarrow$ (g).
\end{proposition}

Before we present the proof and illustrate how to apply this proposition to get Theorem \ref{thm:sta}, we wish to provide some intuitions of this seemingly odd proposition. One should view $x$ and $\hat x$ as the original and recovered signals respectively. Then conditions (b) and (c) are consequences of 
\begin{equation}\label{equ:lrho}
\hat x=\arg\min\|\bar x\|_1\quad\quad \text{subject to }\quad\quad \|A\bar x -y\|_2\leq\rho.
\end{equation} But practically we shall view (b)+(c) being equivalent to \eqref{equ:lrho} because in the literature of stability analysis with $\ell_1$ minimization, as far as we know, the proof only utilizes conditions (b) and (c)  when model \eqref{equ:lrho} is used. For example, see proofs of \cite[Theorem 1.2]{RIP08},  \cite[Theorem 2.1]{CZ14}, etc. 

With this in mind, (a)+(b)+(c) $\Rightarrow$ (d) is the classical statement saying that RIP guarantees a good stability result with $\ell_1$ minimization (Theorem \ref{thm:cai}), except here we allow other norm estimations in (d). This proposition is saying that if certain RIP conditions guarantees a stability result (with $\ell_1$ minimization), then the RIP-NSP condition will also guarantee a similar stability result. Moreover, the RIP-NSP constant matches the RIP constant.

As the main application of Proposition \ref{pro:main}, we see that Theorem \ref{thm:cai} states that (a)+(b)+(c) $\Rightarrow$ (d)  (again in the proof the only information extracted from \eqref{equ:lrho} is (b) and (c)) with $\rho=\epsilon$, $A=\Phi$, $\|\cdot\|_*=\|\cdot\|_2$, and $\delta_{2s}<\frac{1}{\sqrt{2}}$,  and therefore we have the following corollary.

\begin{corollary}\label{lm:cm}
Suppose $\Phi$ is $M\times N$ and has RIP-NSP with $\delta_{2s}<\frac{1}{\sqrt 2}$. If $\hat x, x$ satisfies $\|\hat x\|_1\leq\|x\|_1$ and $\|\Phi \hat x-\Phi x\|_2\leq 2\epsilon$, then
\begin{equation}\label{eq:RIPclass}\|\hat x-x\|_2\leq D_1\frac{\sqrt{1+\delta_{2s}}}{\lambda(\Phi)}\sqrt{\frac{N}{s}}\epsilon+D_2\frac{\sigma_s(x)}{\sqrt{s}}
\end{equation}
\end{corollary}

Rewriting Corollary \ref{lm:cm} in the language of model \eqref{equ:le} yields Theorem \ref{thm:sta}.

\begin{proof}[Proof of Proposition \ref{pro:main}]
The assumption (e) implies that there exists an invertible matrix $U$ such that $\Phi=UA$ and $A$ has property RIP with $\delta_{2s}$. Moreover,
\begin{equation}\notag
\|Ax-A\hat x\|_2=\|U^{-1}\Phi x-U^{-1}\Phi\hat x\|_2\leq \frac{1}{\lambda(U)}\|\Phi x-\Phi\hat x\|_2\leq \frac{2\epsilon}{\lambda(U)}.
\end{equation}

At this point, assumptions (a) (b) and (c) are satisfied with $\rho=\frac{\epsilon}{\lambda(U)}$, therefore we have
\begin{equation}\label{equ:error}
\|\hat x-x\|_*\leq C_1\frac{\epsilon}{\lambda(U)}+C_2\sigma_s(x).
\end{equation}
The rest of the proof is to estimate $\lambda(U)$. Note that

$\lambda^2(\Phi)=\min_{\|x\|=1}\langle UAA^*U^*x,x\rangle=\min_{\|x\|=1}\langle AA^*U^*x,U^*x\rangle\leq\|AA^*\|_{op}\min_{\|x\|=1}\|U^*x\|^2=\|A\|_{op}^2\lambda(U)^2$. Therefore, with Lemma \ref{lem:op},
\begin{equation}\label{equ:lambda}\lambda(U)\geq \lambda(\Phi)/\|A\|_{op}\geq\frac{\lambda(\Phi)}{\sqrt{\frac{N}{2s}+1}\sqrt{1+\delta_{2s}}} .
\end{equation}

Plug \eqref{equ:lambda} into \eqref{equ:error}, we get the desired error bound (g).
\end{proof}

\begin{remark}
The proof only uses the RIP property implicitly in terms of Lemma \ref{lem:op}. Therefore Proposition \ref{pro:main} could be stated more generally as replacing RIP by any property that produces a bound like \eqref{equ:boundphi}.
\end{remark}

\subsection{RIP-NSP is strictly stronger than NSP}

This is done by Theorem \ref{thm:sta} with explicit construction. The idea is to construct a matrix with NSP, but does not satisfy the bound \eqref{equ:ripnsp}, therefore does not satisfy RIP-NSP. This says that NSP is strictly weaker than RIP-NSP.

\begin{lemma}[Candes~\cite{RIP08}]\label{lm:candes}
Assume $\Phi\in \mathbb{R}^{M\times N}$ satisfies RIP of order $2s$ with constant $\delta_{2s}\in (0,1)$. Then $\Phi$ has the NSP of order $s$ with constant $\gamma=\frac{\sqrt{2}\delta_{2s}}{1-\delta_{2s}}$.
\end{lemma}

\begin{lemma}[Rudelson et al. \cite{RV08}]\label{lm:rv}
For any fixed absolute  constant $\delta<1$, when $M,N,s$ satisfy $M\geq C_1t s \log ^4 N$, the random partial Fourier matrix formed by randomly choosing $M$ rows from an $N\times N $ DFT matrix satisfies the RIP with constant $\delta_{2s}\leq \delta $ with probability $1- 5e^{-\frac{\delta^2t}{C_2}}$, where $C_1,C_2$ are absolute constants.

\end{lemma}
The result in \cite{RV08} actually only directly proves the lemma for $\delta=0.5$, but the same proof works for general $\delta$. Moreover, the order of $M$ in Lemma \ref{lm:rv} has been slightly improved recently by Bourgain \cite{B14}, and subsequently by Haviv et al. \cite{HR15}. The interested readers can thus obtain an improved version of Theorem \ref{thm:no}. But we still choose to include the original Rudelson and Vershynin result here since this improvement is not the main point of this paper.

In this paper, we work with real vector spaces due to the restriction of Theorem \ref{thm:cai}, so we need to have a lemma that can turn complex matrices to real matrices while preserving the restricted isometry property.

\begin{lemma}\label{lm:real}
If $F=A+Bi\in\C^{M\times N}$ has RIP with $\delta_s$, then so is the real matrix $R=\left[\begin{array}{c}A\\ B\end{array}\right]$.
\end{lemma}
\begin{proof}
By the definition of RIP, we need to prove for every index set $T\in \{1,...,M\}$ with cardinality $s$, it holds that \begin{equation*}
(1-\delta_s)\|x\|_2^2\leq\|R_Tx\|_2^2\leq (1+\delta_s)\|x\|_2^2.
\end{equation*}
for any $x\in \R^{s} $.
But this is true if we apply RIP of $F$ on $x$ and the fact that $\|R_Tx\|_2^2=\|A_Tx\|_2^2+\|B_Tx\|_2^2=\|F_Tx\|_2^2$.
\end{proof}
\begin{corollary}\label{lm:A}
For any given positive constant $\gamma<1$, there exists $C(\gamma) $, such that when $M,s,N$ satisfy $N/2\geq M \geq C(\gamma) s\log ^4 N$, then there exists a matrix $R\in \R^{M \times N}$ that has $(s, \gamma)$-NSP and $(1,...,1)^T\in \ker (R)$.
\end{corollary}
\begin{proof}
If $M$ is even, then the partial DFT matrix  $F\in\C^{\frac{M}{2}\times N}$ satisfies RIP $\delta_{2s}\leq\delta$ with probability at least $p_1=1-5e^{\frac{-\delta^2t}{C_2}}$ if  $M\geq 2C_1t s \log ^4 N$. We further require $(1,...,1)^T\in \ker (F)$, which is equivalent to not selecting the first row of the DFT matrix. Probability of not selecting the first row $p_2$ is at least $(N-M/2)/N\geq 1/2$. The existence of $F$ such that it has RIP with $\delta_{2s}\leq\delta$ and 
$(1,...,1)^T\in \ker (F)$ is guaranteed if $p_1+p_2\geq1$. This can be done if we set $p_1=1/2$, which requires $M\geq C(\delta)s\log^4 N$, where $C(\delta)=2\log 10 C_1C_2/\delta^2$.

By Lemma \ref{lm:real}, the matrix $R=\left[\begin{array}{c}\text{Real}(F)\\ \text{Im}(F)\end{array}\right]\in\R^{M\times N}$ has RIP with $\delta_{2s}<\delta$ and $(1,...,1)^T\in \ker (R)$. By Lemma \ref{lm:candes}, we get the desired matrix if we set $\delta=\frac{\gamma}{\sqrt 2+\gamma}$, which requires $M\geq C(\gamma) s\log^4N,$ with 
$$C(\gamma)= 2\log 10 C_1C_2\left(\frac{\sqrt{2}+\gamma}{\gamma}\right)^2.$$

If $M$ is odd, we can do exactly the above with the even number $M-1$. So we can get $R'\in\R^{(M-1)\times N}$ that has $(s,\gamma)$-NSP and $(1,...,1)^T\in \ker (R')$. The desired $R$ can be obtained by adding one row that is orthogonal to $(1,\cdots,1)^T$ to $R'$.

\end{proof}

%
%

\begin{proof}[Proof of Theorem \ref{thm:no}]

\textbf{Step 1: Construction of $\Phi$}

The inequality $M\geq 2C(\gamma/3) s\log^4N$ implies $(M-s)\geq C(\gamma/3) s\log^4(N-s)$, so we can find $A\in\R^{(M-s)\times(N-s)}$ that has $(s,\gamma/3)$-NSP and  $e=(1,...,1)^T\in\R^{N-s}\in\ker A$. Define  $$\Nc_e=\{x\in \ker(A): x\perp e\},$$ and let
$$\Nc_e'=\{(\underbrace{0,...,0}_{s},x):x\in \Nc_e\},$$ $$\mathcal{N}=\Nc_e'\oplus \text{span}(d),$$
where $d=\left(\frac{N-4s}{2^2}\gamma,\frac{N-4s}{2^3}\gamma,...,\frac{N-4s}{2^{s+1}}\gamma,\underbrace{-1,....,-1}_{N-s}\right)$.
In addition, if we set $$\varphi_1=\frac{1}{\rho}(\underbrace{\alpha,...,\alpha}_s, \gamma,...,\gamma)$$ with $\alpha=\frac{2(N-s)}{(N-4s)(1-2^{-s})}$ and $\rho$ be a normalization constant such that $\|\varphi_1\|_2=1$, then it is easy to verify that $\varphi_1 \perp \mathcal{N}$, and therefore there exists an orthonormal basis  of $\mathcal{N}^{\perp}$ with $\varphi_1$ being one of the basis vector.
Let $\Phi$ be the matrix whose rows are this orthonormal basis, where $\varphi_1$ is the first row. 

Some quick facts about $\Phi$:
\begin{itemize}
\item $\ker\Phi=\mathcal{N}$.
\item $\Phi\Phi^*=I$, and therefore $\lambda(\Phi)=1$.
\end{itemize}

\noindent\textbf{Step 2: $\Phi$ has $(s,\gamma)$-NSP}

Since $\mathcal{N}=\Nc_e'\oplus d$, we only need to prove any $b=h+d$ satisfies NSP, where $h\in\Nc_e'$. Let $I=\{1,...,s\}$, and we first make two observations: 
\begin{enumerate}
\item[1).] $\|d_I\|_1=\frac{N-4s}{2}\left(1-\frac{1}{2^s}\right)\gamma < \frac{(N-s)\gamma}{2}=\frac{\gamma}{2}\text{sum}(-d_{I^c})=\frac{\gamma}{2}\text{sum}(-d_{I^c}-h)\leq \frac{\gamma}{2}\|d_{I^c}+h\|_1=\frac{\gamma}{2}\|d_{I^c}+h_{I^c}\|_1$, where $\text{sum}(x)$ is the sum of all coordinates of $x$.
\item[2).] $b_{I^c}=h_{I^c}+d_{I^c}\in\ker A$, so if an index set $S\subset I^c$, then $$\|b_S\|_1=\|(b_{I^c})_S\|_1\leq\gamma/3\|(b_{I^c})_{S^c}\|_1=\gamma/3 \|b_{S^c\cap I^c}\|_1.$$
\end{enumerate}

Let $T \in \{1,...,N\}$ be any index set with cardinality $s$, then the above implies
\begin{align*}
\|b_T\|_1&=\|b_{T\cap I}\|_1+\|b_{T\cap I^c}\|_1\leq\|b_I\|_1+\|b_{T\cap I^c}\|_1 = \|d_I\|_1+\|b_{T\cap I^c}\|_1 \\
&< \frac{\gamma}{2} \|(d+h)_{I^c}\|_1+ \|b_{T\cap I^c}\|_1
= \frac{\gamma}{2}\|b_{I^c}\|+\|b_{T\cap I^c}\|_1\\
&=\frac{\gamma}{2}\|b_{I^c\cap T^c}\|+\frac{\gamma}{2}\|b_{I^c\cap T}\|+\|b_{T\cap I^c}\|_1=\frac{\gamma}{2}\|b_{I^c\cap T^c}\|_1+\frac{2+\gamma}{2}\|b_{I^c\cap T}\|_1\\
& \leq \frac{\gamma}{2} \|b_{I^c\cap T^c}\|_1+\frac{(2+\gamma)\gamma/3}{2} \|b_{I^c\cap T^c}\|_1 \leq  \gamma \|b_{I^c\cap T^c}\|_1\leq \gamma \|b_{T^c}\|_1,
\end{align*}

\vspace{0.1in}

\noindent\textbf{Step 3: Recovery with $\Phi$ does not satisfy \eqref{eq:RIPclass}}

Let $x_0=\left( \frac{N-4s}{2^2}\gamma^2, \frac{N-4s}{2^3}\gamma^2,...,\frac{N-4s}{2^{s+1}}\gamma^2, 0,...,0\right)^T$  and $z=\left(\rho,0,...,0\right)^T$. We wish to recover $x_0$ from \eqref{equ:le} with $y=\Phi x_0-z$ and $\epsilon=\|z\|_2=\rho$.

Notice $$\hat x:=x_0-\rho\varphi_1-\gamma d=(\underbrace{-\alpha,-\alpha,....,-\alpha}_s,0,...0)^T$$ is feasible in the problem \eqref{equ:le} because
$$\|\Phi\hat x-y\|_2=\|\Phi x_0-\Phi\rho\varphi_1-\Phi d-\Phi x_0+z\|_2=\|-\rho\Phi\varphi_1+z\|_2=0.$$ The last equality holds because the rows of $\Phi$ are orthogonal and $\varphi_1$ is the first row of $\Phi$.

We wish to get an error bound using Corollary  \ref{lm:cm}, so we need to compare their $\ell_1$ norm.
 
$$\|\hat x\|_1=\alpha s=\frac{2s(N-s)}{(N-4s)(1-2^{-s})}, \|x_0\|_1=\frac{1}{2}(N-4s)(1-2^{-s})\gamma^2.$$

In order to show $\|\hat x\|_1\leq \|x_0\|_1$, it suffices to show that $(N-4s)^2(1-2^{-s})^2\gamma^2>4(N-s)s$, which is true under the assumption that $N\geq 11s/\gamma^2$ and $s\geq 5$.

If $\Phi$ were RIP-NSP with constant $\delta_{2s}$, then by Corollary \ref{lm:cm}, we would get
\begin{equation}\label{equ:bound}
\|\hat{x}-x_0\|_2\leq D_1\sqrt{1+\delta_{2s}} \sqrt {\frac{N}{s}}\frac{\epsilon}{\lambda(\Phi)}=D_1\sqrt{1+\delta_{2s}} \sqrt{\frac{N(N\gamma^2+(\alpha^2-\gamma^2)s)}{s}}
\end{equation}

On the other hand, we can compute that
\begin{align}\label{equ:lower}\notag\|\hat x-x_0\|_2^2&=\|\rho\varphi_1+\gamma d\|^2_2=\sum_{i=1}^s\left(\alpha+\frac{N-4s}{2^{i+1}}\gamma^2 \right)^2\\ 
&=\sum_{i=1}^s(\alpha^2+2\alpha\frac{N-4s}{2^{i+1}}\gamma^2+\frac{(N-4s)^2}{4^{i+1}}\gamma^4 )\geq \frac{(N-4s)^2}{16}\gamma^4.
\end{align}

Under the assumption that $N\geq 11s/\gamma^2$ and $s\geq 5$, we can estimate that $2<\alpha<3$.

Let the right hand side of \eqref{equ:bound} be $R$, so
\begin{equation}\label{equ:upper}
\|\hat x-x_0\|_2^2\leq R^2= C\frac{N(N\gamma^2+(\alpha^2-\gamma^2)s)}{s}\leq \frac{2C\gamma^2N^2}{s},
\end{equation} where $C=\frac{8(1+\delta_{2s})^2}{(1-\sqrt{2}\delta_{2s})^2}$ and we used the relation $N\geq 11s/\gamma^2$.

If $s>80C/\gamma^2$, then we arrive at a contradiction between \eqref{equ:lower} and \eqref{equ:upper}.

\end{proof}




\end{document}